%% file: format1.tex
\documentclass[letterpaper, 10 pt, conference]{ieeeconf}  

\IEEEoverridecommandlockouts                              
\usepackage{amsfonts,amsmath,amssymb}
\usepackage{graphicx}

\usepackage{ifthen}
\usepackage{color}
\usepackage{cite}

\usepackage{amsfonts}
\usepackage{amsmath}
\usepackage{amssymb}
\usepackage{color}
\usepackage{graphics}
\usepackage{algorithmicx}
\usepackage{algpseudocode}
\usepackage[boxruled,vlined,linesnumbered,onelanguage]{algorithm2e}
\usepackage{caption}
\usepackage{epsfig}
\usepackage{chngcntr}
\usepackage{dashrule}
\usepackage[position=b]{subcaption}
\usepackage[hyperfootnotes=false]{hyperref}
\usepackage{listings}
\usepackage{color}
\usepackage{authblk}

\usepackage[multiple]{footmisc}

\usepackage{xspace}
\usepackage{tikz}
\usepackage{subcaption}
\usetikzlibrary{calc, positioning, fit, arrows, arrows.meta, shapes.geometric, shapes.symbols, shapes.misc, patterns}

\def\scr#1{{\cal #1}}
\newtheorem{theorem}{Theorem}
\newtheorem{lemma}{Lemma}

\def\qed{ \rule{.08in}{.08in}}

\input{steve.tex}

\newcommand{\drawnodex}[6]{
	\path (#1) node[draw,thick,circle,minimum width=1em,minimum height=1em,inner sep=0.2ex] (#2) {} node[inner sep=0.2ex,fill=white,#3=0 of #2] {#4} node[inner sep=0.2ex,fill=white,#5=0 of #2] {#6};
}

\newcommand{\drawrededge}[2]{
	\draw[ultra thick,red] (#1) -- (#2);
}
\newcommand{\drawgreenedge}[2]{
	\draw[ultra thick,green] (#1) -- (#2);
}

\newcommand{\drawallyx}[5][0.5]{
	\path (#2) -- (#3) node(_1)[allow upside down,sloped,above,minimum height=1.2em,anchor=center,pos=0.3] {};
	\path (#2) -- (#3) node(_2)[allow upside down,sloped,above,minimum height=1.2em,anchor=center,pos=0.7] {};
	\draw[-stealth,ultra thick,green!80!black] (#2) .. controls (_1.north) and (_2.north) .. (#3) node[inner sep=0.2ex,fill=white,sloped,pos=#1] {#4};
	\draw[-stealth,ultra thick,green!80!black] (#3) .. controls (_2.south) and (_1.south) .. (#2) node[inner sep=0.2ex,fill=white,sloped,pos=#1] {#5};
}
\newcommand{\drawfoex}[5][0.5]{
	\path (#2) -- (#3) node(_1)[allow upside down,sloped,above,minimum height=1.2em,anchor=center,pos=0.3] {};
	\path (#2) -- (#3) node(_2)[allow upside down,sloped,above,minimum height=1.2em,anchor=center,pos=0.7] {};
	\draw[-stealth,ultra thick,red] (#2) .. controls (_1.north) and (_2.north) .. (#3) node[inner sep=0.2ex,fill=white,sloped,pos=#1] {#4};
	\draw[-stealth,ultra thick,red] (#3) .. controls (_2.south) and (_1.south) .. (#2) node[inner sep=0.2ex,fill=white,sloped,pos=#1] {#5};
}

\title{Countries' Survival in Networked International Environments
}
\author{Yuke Li, A. Stephen Morse, Ji Liu, and Tamer Ba\c{s}ar\thanks{
This work was supported by National Science Foundation grant n.1607101.00, US Air Force grant n. FA9550-16-1-0290 and US Army Research Office Grant W911NF-16-1-0485.
Y.~Li and A. S.~Morse are respectively with the Department of Political Science and the Department of Electrical Engineering, Yale University (\texttt{\{yuke.li, as.morse\}@yale.edu}). J.~Liu is with the Department of Electrical and Computer Engineering, Stony Brook University (\texttt{ji.liu@stonybrook.edu}).
T.~Ba\c{s}ar is with the Department of Electrical and Computer Engineering,
University of Illinois at Urbana-Champaign (\texttt{basar1@illinois.edu}).
}}

\begin{document}

\maketitle
\thispagestyle{empty}
\pagestyle{empty}

\begin{abstract}
This paper applies a recently developed power allocation game in \cite{allocation} to study the countries' survival problem in networked international environments. In the game, countries strategically allocate their power to support the survival of themselves and their friends and to oppose that of their foes, where by a country's survival is meant when the country's total support equals or exceeds its total threats. This paper establishes conditions that characterize different types of networked international environments in which a country may survive, such as when all the antagonism among countries makes up a complete or bipartite graph.

\end{abstract}


\section{Introduction}

In a recent paper \cite{allocation}, a \emph{power allocation game} (the PAG, hereafter) on networks has been studied to understand countries' strategic behaviors in international environments. In the game, countries allocate their resources, i.e., deploying their total power, to their friends and foes in order to pursue certain goals, such as protecting the survival of their friends and themselves while opposing that of their foes. In \cite{allocation}, the framework of the power allocation game, which is an infinite, resource-allocation game on graphs, is introduced, and the existence of a pure strategy Nash equilibrium is established. 

A question both immediate from studying countries' power allocation in \cite{allocation} and meaningful from a decision making perspective is this: can a country \emph{always} survive when it allocates its power in a complex, networked international environment? This question is motivated by a well-accepted assumption in international relations theory that for a country, its survival is a  fundamental objective that needs to be pursued on an \emph{everyday} basis.  For instance, according to John Mearsheimer, a representative of the school of thought termed as ``offensive realism'' which generally studies the scenarios in which countries are aggressive and even expansionist, countries are seeking nothing more than their own security and survival at a minimum\cite{mearsheimer2001tragedy} (similar discussions can also be found in \cite{waltz1959man} and \cite{waltz1979theory}).

The studies of countries' survivability in complex international environments necessarily fall into the broad category of the studies of military strategies, which can be dated to the work of Sun Tzu \cite{tzu2009art} and Carl von Clausewitz \cite{von1873war}. In terms of the famous examples of contemporary scholarship on military strategies, Jack L Snyder's Rand Report for the US Air Force entitled ``The Soviet strategic culture: Implications for limited nuclear operations''  \cite{snyder1977soviet} explores several factors -- historical, institutional, and political -- that are conducive to a uniquely Soviet approach to strategic thought, and \cite{spykman1942america,weigley1977american,art1991defensible, carter2000preventive, gray1981national,  gacek1994logic} focuses more on the American approach. In addition, the work of  \cite{luttwak2001strategy, murray1994making, schelling1980strategy} gives a more general treatment of military strategies. Outside of academia, \cite{wylie2014military} and \cite{cook2001us} are two studies published within the American military system, with the former developing a theory of ``power control'' and the latter overviewing some major concepts underlying the study of military strategy (e.g., national power, national interests, strategic risk and strategic art). In particular, \cite{schelling1980strategy} is representative of a line of work that applies game theory to explore a list of survival-related issues, including the strategies to force the other side into compliance and the limits of manipulating those strategies (e.g., ``brinkmanship'' by  \cite{nalebuff1986brinkmanship}, ``the spiral model'' by \cite{kydd1997game} and \cite{glaser1992political},  and ``the deterrence model'' by \cite{glaser1992political, zagare1985toward,zagare1987dynamics,kilgour1991credibility,zagare1993asymmetric,zagare1996classical,zagare2000perfect,snyder2015deterrence}).

However, none of the above works have explored countries' survivability within a networked international environment. This paper will take a preliminary step towards applying the ideas recently developed in \cite{allocation} to study this military strategy and national policy of managing the military resources; i.e., the ``power allocation strategies'', in the context of a networked international environment and the implications of this environment for countries' survival.

The remainder of the paper is organized as follows. The theoretical framework of the PAG on networks in \cite{allocation} will first be reviewed in Section II.
Then, Section III will explore two questions: First, what could be the possible environments in which countries can survive when they have friends? Second, what could be the possible environments in which they can survive when they do not have friends? As will be discussed, some results in Section III  characterize the game environments for a given country to survive in at least one equilibrium (class) of the game; i.e., it has any possibility of survival, while other results characterize the game environments for the country to survive in all equilibrium (classes); i.e., it can absolutely survive. Lastly, the environments for uniquely predicting about countries' survival will be discussed in Section IV.  

%
%
%
%
%

\section{The Power Allocation Game}

In this section, the PAG proposed in \cite{allocation} is briefly reviewed, beginning with the definitions of the elements that constitute both the environment in which the PAG it is played and then the PAG itself. In a networked international environment, there is a collection $\mathcal{C}$ of $n$ countries, labeled $1,2, \ldots, n$; let the set of country labels be $\mathbf{n} = \{1,2, \ldots, n\}$.  The total power of all $n$ countries is defined as a real, nonnegative valued row vector $p = [p_{i}]_{1 \times n}$, where $p_{i}$ is country $i$'s total power.


Any two countries in $\mathcal{C}$ can be said to have a \emph{relation}. A \emph{relation} is technically speaking a binary relation defined on $\mathcal{C}$, which takes one of the following four possibilities: with itself, each country $i$ can be said to have a \emph{self} relation; with any other country $j$ in $\mathcal{C}$, each country $i$ can have a \emph{friend}, an \emph{adversary} or a \emph{null} relation. The binary relation is reflexive; in addition, assume that it is symmetric; e.g., if $j$ is a friend of $i$, then $i$ is a friend of $j$ and similarly if $j$ is an adversary of $i$ then $i$ is an adversary of $j$, and that it is not transitive. 

Based on this binary relation, country $i$ has a subset of countries in $\mathcal{C}$ with labels in $\mathcal{F}_{i} \subset \mathbf{n}$ called its \emph{friends}, a subset of countries in $\mathcal{C}$ with labels in $\mathcal{A}_{i} \subset \mathbf{n}$ called its \emph{adversaries}, and the set of countries in $\mathcal{C}$ with labels in $\mathbf{n} - \{i\} \cup \mathcal{A}_{i} \cup \mathcal{F}_{i}$ $i$ has \emph{null} relations with; i.e., having no specific relations with. For each $i \in \mathbf{n}$, $\{i\}, \mathcal{F}_{i}$, $\mathcal{A}_{i}$ are disjoint. 

The unordered pair $\{i,j\}$ stands for a pair of distinct labels in $\mathbf{n}$ such that $i$ and $j$ have a relation.  Denote the set of all friendly pairs as $\mathcal{R}_{f}$ and the set of all adversarial pairs as $\mathcal{R}_{a}$. Suppose the number of pairs in $\mathcal{R}_{f} \cup \mathcal{R}_{a}$ is $m$. A map $\eta: \mathcal{R}_{f} \cup \mathcal{R}_{a} \to \mathbf{m}$ where $\mathbf{m} = \{1,2, ..., m\}$ determines for each element in $\mathcal{R}_{f} \cup \mathcal{R}_{a}$ a distinct label in the set $\mathbf{m}$. 

Country $i$'s power allocation strategy is a real, nonnegative valued row vector $u_{i} \in \mathbb{R}^{1 \times n}$ whose $j$-th entry is $u_{ij}$. If $j \in \mathcal{F}_{i}$, then $u_{ij}$ represents the portion of country $i$'s total power which country $i$ is willing to commit to the support or defense of friend $j$ against friend $j$'s adversaries. If $j \in \mathcal{A}_{i}$, then $u_{ij}$ is the portion of country $i$'s total power that it is committing to its possible offense actions against country $j$. If $j \in \{i\}$, $u_{ii}$ is the portion of country $i$'s total power it holds in reserve. Finally, if $\mathbf{n} - \{i\} \cup \mathcal{A}_{i} \cup \mathcal{F}_{i}$, $u_{ij}$ represents the portion of country $i$'s total power committing to $j$ which $i$ has no specific relation with, and we stipulate that $u_{ij} = 0$. Accordingly, for each $i \in \mathbf{n}$, $\sum_{j = 1}^{n}u_{ij} = p_{i}$ so the $i$-th row sum of the \emph{power allocation matrix } $U = [u_{ij}]_{n \times n}$ is $p_{i}$. $\mathcal{U}$ denotes the set of all admissible strategy matrices.

For each country $i \in \mathbf{n}$, there are two types of nonnegative-valued functions on $\mathcal{U}$.  The first, called a \emph{support function} for agent $i$, is the map $\sigma_{i}$: $\mathcal{U} \to [0, \infty)$,
\begin{equation*}
U \longmapsto u_{ii} + \sum_{j\in \mathcal{F}_{i}}u_{ji} + \sum_{j \in \mathcal{A}_{i}}u_{ij}
\end{equation*} Here $\sum_{j\in \mathcal{F}_{i}}u_{ji}$ is the total amount of power the friends of country $i$ commit to country $i$'s defense and $\sum_{j \in \mathcal{A}_{i}}u_{ij}$ is the total amount of power country $i$ commits to its possible offenses against all of its adversaries. The second function, called a \emph{threat function} for country $i$, is the map $\tau_{i}: \mathcal{U} \to [0, \infty)$,

\begin{equation*}
U \longmapsto \sum_{j \in \mathcal{A}_{i}}u_{ji}
\end{equation*}Thus $\tau_{i}(U)$ is the total power of all of country $i$'s adversaries commit to their respective offenses against country $i$.

As a consequence of specific allocations, each country $i$ may find itself in one of three possible states, namely a \emph{safe} state, a \emph{precarious} state, or an \emph{unsafe} state. A country is said to \emph{survive} if it is in the safe or precarious state. Let $x_{i} : \mathcal{U} \to \{\text{safe}, \text{precarious}, \text{unsafe}\}$ denote the map

\begin{equation*}
U \longmapsto \begin{cases}
\text{safe}   & \text{if}~ \sigma_{i}(U) > \tau_{i}(U)\\
\text{precarious}  & \text{if}~ \sigma_{i}(U) = \tau_{i}(U)\\
\text{unsafe} & \text{if}~ \sigma_{i}(U) < \tau_{i}(U)

\end{cases}
\end{equation*}

$x_{i}(U)$ is called the state of country $i$ induced by power allocation matrix $U \in \mathcal{U}$. More generally, by the state of the overall collection of countries $\mathcal{C}$ induced by power allocation matrix $U$ is meant the row vector $x(U) = [x_{i}(U)]_{1 \times n}$. The state space of $\mathcal{C}$ is thus the finite set $\mathcal{X} = \{x(U) : U \in \mathcal{U}\}$ whose cardinality is  at most $3^{n}$.

In the sequel, two axioms are proposed to construct a rationale for countries to choose their own power allocation strategies in a game-theoretic context, which are based on the states of themselves, their friends and adversaries induced by the power allocation matrices. The axioms will model their \emph{preferences} for all possible strategy combinations of all countries; i.e.,  the power allocation matrices in $\mathcal{U}$. 

``Preference'' is a terminology commonly used in the social sciences to describe agents' ordering of alternatives; take two arbitrary power allocation matrices $U$ and $V$, country $i$ may have a strong preference relation, e.g.,``strongly prefer''\footnote{In many other contexts, it is also termed as ``strictly prefer''.} $V$ to $U$ (written as $U \prec V$), a weak preference relation, e.g., ``weakly prefer'' $V$ to $U$ (written as $U \preceq V$), or an indifference relation, i.e., are indifferent between the two (written as $U \sim V$).

It is natural to presume that any country $i$ cares positively about the survival of its friends and itself, negatively about the survival of its adversaries, and indifferently about the survival of those countries with whom it has no relations. These observations motivate Axiom 1: 

{\bf Axiom 1 (Multi-Front Survival Issue):} 

1. Country $i$ \emph{weakly prefers} strategy matrix $V$ over $U$ if the following two conditions both hold: 

\begin{description}
\item[a)] $(x_{j}(V) \in \{\text{safe}, \text{precarious}\})$ or $(x_{j}(U) = \text{unsafe})$ or both, $\forall j \in \{i\} \cup \mathcal{F}_{i}$\footnote{For convenience, it can equivalently be written as $(x_{j}(V) \in \{\text{safe}, \text{precarious}\}) \lor (x_{j}(U) = \text{unsafe})$, $\forall j \in \{i\} \cup \mathcal{F}_{i}$.}.
\item[b)] $(x_{j}(V) \in \{\text{unsafe}, \text{precarious}\})$ or $(x_{j}(U) = \text{safe})$ or both, $\forall j \in \mathcal{A}_{i}$.
\end{description}
As is standard, weak preference of $V$ over $U$ is denoted by $V \preceq U$ or by $U \preceq V$.

2.  Country $i$ is \emph{indifferent} to the choice between strategy matrices $V$ and $U$ if  $$x_{j}(U) = x_{j}(V), \forall j \in \{i\} \cup \mathcal{A}_{i} \cup \mathcal{F}_{i}.$$
Indifference between $U$ and $V$ is denoted by $U \sim V$.

A direct implication of Axiom 1 is that country $i$ weakly prefers $V$ over $U$ if country $j$'s states\footnote{$j$ may or may not have a relation with $i$.} induced by $U$ and $V$ satisfy one of the three respective conditions above, while all else is equal (i.e., the states of all (other) countries $i$ has a relation with are the same under $U$ and $V$).

It is also natural to assume that countries prioritize self-survival. This motivates Axiom 2:

{\bf Axiom 2 (Priority of Self-Survival):} 
Country $i$ \emph{strongly prefers} strategy matrix $V$ over $U$ if  $$(x_{i}(V) \in \{\text{safe}, \text{precarious}\}) ~\text{and}~ (x_{i}(U) = \text{unsafe}).\footnote{Similarly, it can be equivalently written as $(x_{i}(V) \in \{\text{safe}, \text{precarious}\}) \land (x_{i}(U) = \text{unsafe})$.}$$
Strong preference of $V$ over $U$ is denoted by $V \prec U$ or by $U \prec V$.

The two axioms determine \emph{a partial order} of the states in $\mathcal{X}$. An additional assumption about ad-hoc and country-specific attributes such as degrees of affinity with different friends (e.g., captured by cultural, trading, linguistic connections) is both necessary and sufficient for extending this partial order into a total order; a discussion of why this is so will appear in another paper. Moreover, Axiom 2 makes this partial order a ``lexicographic order'', with the further implication being that a continuous, real-valued utility function representation of the preference order of the states in $\mathcal{X}$ is impossible (see the discussion in \cite{sen2014collective}).



The PAG in a networked international environment is the collection of all the aforementioned elements, $\Gamma = \{\mathcal{C},  p, \mathcal{U}, \sigma_{i}, \tau_{i}, \mathcal{X}, \preceq\}$.  

For illustrative purpose, an edge-colored, undirected and unweighted graph on $n$ vertices and $m$ edges, $\mathbb{G}_{E} = (\mathcal{V}, \mathcal{E}_{E})$, is the ``environment graph'' of the PAG. An \emph{environment graph} represents $n$ countries with their total power labeled besides each of the $n$ vertices, $m$ pairs of which have a friend or adversary relation ($n, m \in \mathbb{Z}$). In $\mathbb{G}_{E}$, two nodes $v_{i}$ and $v_{j}$, which denote two countries $i$ and $j$ with a friend or adversary relation, are connected by an undirected edge $\{v_{i},v_{j}\}$, colored green if $i$ and $j$ are friends and red if $i$ and $j$ are adversaries. 

An edge-colored, directed and weighted graph on $n$ vertices and $2m$ edges, $\mathbb{G}_{A} = (\mathcal{V}, \mathcal{E}_{A})$, is the ``allocation graph'' of the PAG. An \emph{allocation graph} represents the power allocation of countries in this environment; i.e., a power allocation matrix. In $\mathbb{G}_{A}$, two nodes $v_{i}$ and $v_{j}$, which denote two countries $i$ and $j$ with a friend or adversary relation, are connected by two directed edges $(v_{i},v_{j})$ and $(v_{j},v_{i})$, colored green if $i$ and $j$ are friends and red if $i$ and $j$ are adversaries. The edge weight of $(v_{i},v_{j})$ is $u_{ij}$, and the node weight of $i$ is $u_{ii}$. Neither $\mathbb{G}_{E}$ nor $\mathbb{G}_{A}$ has to be connected. When they are unconnected, the PAGs on those connected components can be regarded as being unrelated.

The Nash equilibrium concept is naturally employed to make predictions for the PAG. Let country $i$'s deviation from the power allocation matrix $U$ be a nonnegative-valued $1\times n$ row vector $d_{i} \in \mathbb{R}^{1 \times n}$ such that $u_{i} + d_{i}$ is a valid strategy that satisfies the total power constraint for country $i$. The deviation set $\mathcal{\delta}_{i}(U)$ is the set of all possible deviations of country $i$ from the power allocation matrix $U$. In the context of a PAG, a power allocation matrix $U$ is a pure strategy Nash Equilibrium if no unilateral deviation in strategy by any single country $i$ is profitable for $i$, that is, 
\begin{equation*}
U + e_{i} d_{i} \preceq U, \;\;\;\;\; {\rm for\; all} \;\;\;\;\; d_{i} \in \mathcal{\delta}_{i}(U),
\end{equation*} 
where $e_{i}$ is an $n \times 1$ unit vector whose elements are $0$ but the $i$-th coordinate which is $1$.

Denote the set of pure strategy Nash equilibria as $\mathcal{U}^{*}$. An equivalence relation can be defined on $\mathcal{U}^{*}$ such that $U^{*}$ is equivalent to $V^{*}$ if and only if $x(U^{*}) = x(V^{*})$. Let $[\mathcal{U}^{*}]_{x(U^{*})}$ be the \emph{equilibrium equivalence class} of $U^{*} \in \mathcal{U}^{*}$. Obviously, the total number of equilibrium equivalence classes is at most $3^{n}$, and their union is  $\mathcal{U}^{*}$.

It has been established in \cite{allocation} that the PAG $\Gamma = \{\mathcal{C}, p, \mathcal{U}, \sigma_{i}, \tau_{i}, \mathcal{X}, \preceq\}$ always has a pure strategy Nash equilibrium.

\section{Countries' Survival}

This section explores the conditions for countries' survival, by discussing the first case in which they have friends in the environment, and then the second one in which they do not.

\subsection{Survival with Friends}

This section shows that a country's friends' outside obligations may tremendously affect its survival.

\subsubsection*{Example 1} The networked environment in which the PAG takes place is characterized by the following parameters:

\begin{enumerate}

\item Set of countries (labels): $\mathbf{n} = \{1,2,3, 4,5,6\}$.

\item Countries' total power:
$p = [p_{1}, p_{2}, p_{3}, p_{4}, p_{5}, p_{6}] = [19,3,6,15,3,9]$.

\item Countries' relations: $\mathcal{A}_{1} = \{4\}$, $\mathcal{A}_{2} = \{5\}$, $\mathcal{A}_{3} = \{6\}$, $\mathcal{F}_{2} = \{3\}$, $\mathcal{F}_{4} = \{5\}$, and all other possible pairs of countries have no relations.

\item Countries' preferences: Assume the two axioms.

\end{enumerate}

\begin{figure}[ht]

	\begin{subfigure}[t]{0.45\linewidth}
	\centering
	\begin{tikzpicture}[x=4em,y=-4em]
		\drawnodex{0,0}{v1}{left}{$v_{1}$}{above}{$19$}
		\drawnodex{0,1}{v2}{left}{$v_{2}$}{below}{$3$}
		\drawnodex{0,2}{v3}{left}{$v_{3}$}{below}{$6$}
		\drawnodex{2,0}{v4}{above}{$v_{4}$}{below}{$15$}
		\drawnodex{2,1}{v5}{below}{$v_{5}$}{above}{$3$}
		\drawnodex{2,2}{v6}{above}{$v_{6}$}{below}{$9$}
		\drawgreenedge{v2}{v3}
		\drawrededge{v3}{v6}
		\drawrededge{v2}{v5}
		\drawrededge{v1}{v4}
		\drawgreenedge{v4}{v5}
	\end{tikzpicture}\qquad
	\caption{Environment}
	\label{fail1}
	\end{subfigure}
   \begin{subfigure}[t]{0.45\linewidth}
	\centering
	\begin{tikzpicture}[x=4em,y=-4em]
		\drawnodex{0,0}{v1}{left}{$v_{1}$}{above}{$0$}
		\drawnodex{0,1}{v2}{left}{$v_{2}$}{below}{$0$}
		\drawnodex{0,2}{v3}{left}{$v_{3}$}{below}{$0$}
		\drawnodex{2,0}{v4}{right}{$v_{4}$}{above}{$0$}
		\drawnodex{2,1}{v5}{below}{$v_{5}$}{above}{$0$}
		\drawnodex{2,2}{v6}{above}{$v_{6}$}{below}{$0$}
		\drawallyx{v2}{v3}{$0$}{$0$}
		\drawfoex{v3}{v6}{$6$}{$9$}
		\drawfoex{v2}{v5}{$3$}{$3$}
		\drawfoex{v1}{v4}{$19$}{$15$}
		\drawallyx{v4}{v5}{$0$}{$0$}
	\end{tikzpicture}
	\caption{Allocation}
		\label{fail2}
	\end{subfigure}
    \caption{Friends' Outside Commitments}
    \label{fail}
\end{figure}
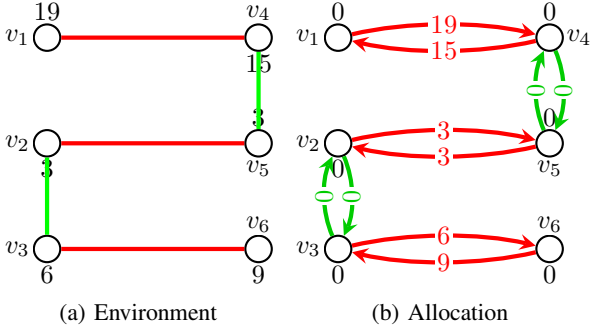

Figure~\ref{fail1} illustrates a two-sided conflict, where countries from one side are either adversaries or not having specific relations with countries from the other side. Figure~\ref{fail2} shows the (only) equilibrium of the corresponding PAG where neither country 3 nor country 4 benefits from its friend relation -- their friends, country 2 and country 4 respectively, have to use all of their power for protecting themselves. Therefore, it is always necessary to consider a country's friends' own power and relations in order to understand to what extent the power of the friends are actually available for supporting itself.

This example motivates two solutions to the ineffectiveness of alliances by which countries may achieve survival in any equilibrium of the PAG, formalized in Theorems~\ref{prop/balance} and \ref{prop/clique}, respectively.

%
%

The following result establishes a sufficient condition under which a group of countries can survive, assuming they are not adversarial with each other.

\begin{theorem}
\label{prop/balance} In the PAG $\Gamma = \{\mathcal{C},  p, \mathcal{U}, \sigma_{i}, \tau_{i}, \mathcal{X}, \preceq\}$, a nonempty set of countries $\mathbf{n}_{0} \subset \mathbf{n}$ with no adversary relation among them will survive in any Nash equilibrium if for each of them, its total power is no smaller than that of all its adversaries.  
In other words, if for each $i \in \mathbf{n}_{0}$, there holds 
$$\mathcal{A}_{i} \bigcap \mathbf{n}_{0} = \emptyset, \;\;\;\;\;
p_{i} \geq \sum\limits_{j\in \mathcal{A}_{i}}p_{j},$$ 
then $\sigma_{i}(U^{*}) \geq \tau_{i}(U^{*})$ for all $i\in\mathbf{n}_0$ and $U^{*} \in \mathcal{U}^{*}$.
\end{theorem}

It is worth noting that Theorem~\ref{prop/balance} holds regardless of whether the countries in $\mathbf{n}_0$ have friends or not.

%
%
%


{\em Proof:}
Let $U^{*} \in \mathcal{U}^{*}$. Consider an arbitrary country $i\in\mathbf{n}$. 
We consider two cases $\sigma_{i}(U^{*}) < p_{i}$ and $\sigma_{i}(U^{*}) \geq p_{i}$, separately. 

First, suppose that $\sigma_{i}(U^{*}) < p_{i}$. 
We claim that $\sigma_{i}(U^{*}) \ge \tau_{i}(U^{*})$. To establish the claim, suppose that, to the contrary, $\sigma_{i}(U^{*}) < \tau_{i}(U^{*})$. Since $\sigma_{i}(U^{*}) < p_{i}$, it implies that country $i$ allocates its power to support its friend(s) while being  unsafe itself, which contradicts Axiom~2. Thus, $\sigma_{i}(U^{*}) \ge \tau_{i}(U^{*})$, i.e., country $i$ can survive in this case.

Next, suppose that $\sigma_{i}(U^{*}) \ge p_{i}$.
Since country $i$'s adversaries are not in the set $\mathbf{n}_{0}$, 
it follows that 
$$\sum_{j \in \mathcal{A}_{i}}p_{j} \geq \tau_{i}(U^{*}).$$ 
Since country $i$'s total power is no smaller than that of all its adversaries, it follows that 
$$\sigma_{i}(U^{*}) \geq p_{i} \geq \sum_{j \in \mathcal{A}_{i}}p_{j} \geq \tau_{i}(U^{*}),$$
which implies that country $i$ can also survive in this case.
\hfill$\qed$

The next theorem treats another scenario in which the group of countries considered are friends with each other.

\begin{theorem}
\label{prop/clique} In the PAG $\Gamma = \{\mathcal{C},  p, \mathcal{U}, \sigma_{i}, \tau_{i}, \mathcal{X}, \preceq\}$, a nonempty set of countries $\mathbf{n}_{0} \subset \mathbf{n}$ will survive in any Nash equilibrium if they are friends with one another and their total power is no smaller than that of their adversaries.  In other words, if
$$j \in \mathcal{F}_{i}, \;\;\;{\rm for \; all} \;\;\; i,j \in \mathbf{n}_{0},$$ and
$$\sum\limits_{i \in \mathbf{n}_{0}}p_{i} \geq \sum\limits_{j\in\mathcal{A}_{\mathbf{n}_0}}p_{j}, \;\;\; \text{where} \;\;\; \mathcal{A}_{\mathbf{n}_0} = \bigcup\limits_{i \in \mathbf{n}_{0}} \mathcal{A}_{i},$$
then $\sigma_{i}(U^{*}) \geq \tau_{i}(U^{*})$ 
for all $i\in\mathbf{n}$ and $U^{*} \in \mathcal{U}^{*}$.
\end{theorem}


{\em Proof:}
Let $U^{*}$ be any Nash equilibrium in $\mathcal{U}^{*}$. 
To prove that all countries in $\mathbf{n}_{0}$ can survive at $U^{*}$ (i.e., $\sigma_{k}(U^{*}) \geq \tau_{k}(U^{*})$ for all $k\in\mathbf{n}_0$), suppose that, to the contrary, there exists at least one country that is unsafe. 
Let $\scr E$ be the set of those unsafe countries. 
That is, $\sigma_{i}(U^{*})  < \tau_{i}(U^{*})$ for each $i\in\scr E$. 
It is clear that $\scr E\subset \mathbf{n}_0$.  
Since the total power of countries in $\mathbf{n}_{0}$ is no smaller than that of their adversaries, it is impossible that all countries in $\mathbf{n}_{0}$ are unsafe. Thus, $\scr E$ must be a proper subset of $\mathbf{n}_0$, which implies that $\mathbf{n}_0\setminus\scr E$ is nonempty.

Since
$$\sum\limits_{i \in \mathbf{n}_{0}}p_{i} \geq \sum\limits_{j\in\mathcal{A}_{\mathbf{n}_0}}p_{j},$$
it follows that 
$$\sum_{i \in \mathbf{n}_{0}\setminus \scr E} p_{i} +  \sum_{i \in \scr E} p_{i} \geq 
\sum\limits_{j \in \mathcal{A}_{\mathbf{n}_0} \setminus \mathcal{A}_{\scr E}}p_{k} + \sum_{j \in \mathcal{A}_{\scr E}}p_{j},$$
where 
$$\mathcal{A}_{\scr E} = \bigcup\limits_{i \in \scr E} \scr A_i.$$
Rearranging the terms of the above inequality, we have   
$$\sum_{i \in \mathbf{n}_{0}\setminus \scr E} p_{i} 
- \sum\limits_{j \in \mathcal{A}_{\mathbf{n}_0} \setminus \mathcal{A}_{\scr E}}p_{k} 
\geq 
\sum_{j \in \mathcal{A}_{\scr E}}p_{j}
- \sum_{i \in \scr E} p_{i}.$$
Since the amount of threats from a country's (or a set of countries') adversaries against the country (or the set of countries) cannot exceed its (or their) total power, it follows that 
$$\sum_{j \in \mathcal{A}_{\scr E}}p_{j} - \sum_{i\in\scr E} p_{i} \geq 
\sum_{i\in\scr E} \left(\tau_{i}(U^{*}) - \sigma_{i}(U^{*})\right).$$
Then, the friends of the countries in $\scr E$ can deviate by transferring at most $\sum_{i\in\scr E} (\tau_{i}(U^{*}) - \sigma_{i}(U^{*}))$ amount of power for supporting the countries in $\scr E$ to avoid being \text{unsafe}, without becoming \text{unsafe} themselves. 
This contradicts the fact that $U^{*}$ is a Nash equilibrium. Thus, all countries in $\mathbf{n}_{0}$ survive at $U^{*}$.
\hfill$\qed$

\subsection{Survival without Friends}

Countries' survival issue in environments where no countries have any friends will now be examined; specifically, two possible networked environments will be discussed, with the first being where the adversary relations constitute a \emph{complete graph}, and the second being where the adversary relations constitute a \emph{bipartite graph}.

\subsubsection*{Example 2}
The networked environment in which the PAG takes place is characterized by the following parameters:
\begin{enumerate}
\item Set of countries (labels): $\mathbf{n} = \{1,2,3\}$
\item Countries' total power: $p = [p_{1},p_{2},p_{3}] = [8, 6,4]$.

\item Countries' relations: $\mathcal{A}_{i} = \mathbf{n}\setminus \{i\}$, $i \in \mathbf{n}$

\item Countries' preferences: Assume the two axioms.

\end{enumerate}

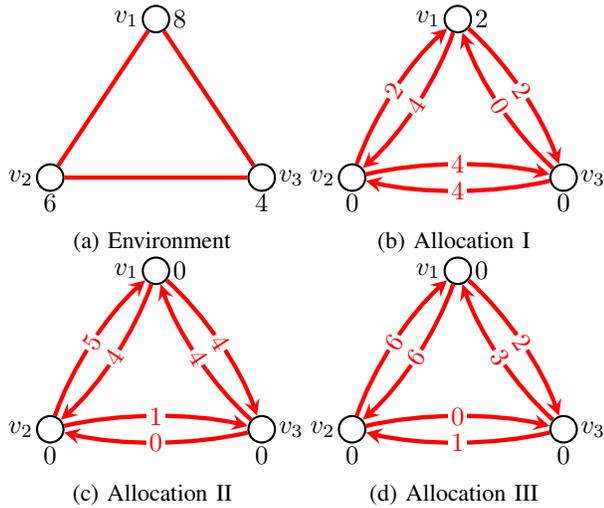
\begin{figure}[ht]
\centering
\begin{subfigure}[t]{0.45\linewidth}
	\centering
	\begin{tikzpicture}[x=4em,y=-4em]
		\drawnodex{0,0}{v1}{left}{$v_{1}$}{right}{$8$}
		\drawnodex{-1,1.5}{v2}{left}{$v_{2}$}{below}{$6$}
		\drawnodex{1,1.5}{v3}{right}{$v_{3}$}{below}{$4$}
				
		\drawrededge{v1}{v3}
		\drawrededge{v2}{v3}
		\drawrededge{v1}{v2}
	\end{tikzpicture}
	\caption{Environment}
	\label{f1}
\end{subfigure}
\begin{subfigure}[t]{0.45\linewidth}
	\centering
	\begin{tikzpicture}[x=4em,y=-4em]
		\drawnodex{0,0}{v1}{left}{$v_{1}$}{right}{$2$}
		\drawnodex{-1,1.5}{v2}{left}{$v_{2}$}{below}{$0$}
		\drawnodex{1,1.5}{v3}{right}{$v_{3}$}{below}{$0$}
				
		\drawfoex{v1}{v3}{$2$}{$0$}
		\drawfoex{v2}{v3}{$4$}{$4$}
		\drawfoex{v1}{v2}{$4$}{$2$}
	\end{tikzpicture}
	\caption{Allocation I}
	\label{f2}
\end{subfigure}
\begin{subfigure}[t]{0.45\linewidth}
	\centering
	\begin{tikzpicture}[x=4em,y=-4em]
		\drawnodex{0,0}{v1}{left}{$v_{1}$}{right}{$0$}
		\drawnodex{-1,1.5}{v2}{left}{$v_{2}$}{below}{$0$}
		\drawnodex{1,1.5}{v3}{right}{$v_{3}$}{below}{$0$}
				
		\drawfoex{v1}{v3}{$4$}{$4$}
		\drawfoex{v2}{v3}{$1$}{$0$}
		\drawfoex{v1}{v2}{$4$}{$5$}
	\end{tikzpicture}
	\caption{Allocation II}
	\label{f3}
\end{subfigure}
\begin{subfigure}[t]{0.45\linewidth}
	\centering
	\begin{tikzpicture}[x=4em,y=-4em]
		\drawnodex{0,0}{v1}{left}{$v_{1}$}{right}{$0$}
		\drawnodex{-1,1.5}{v2}{left}{$v_{2}$}{below}{$0$}
		\drawnodex{1,1.5}{v3}{right}{$v_{3}$}{below}{$0$}
				
		\drawfoex{v1}{v3}{$2$}{$3$}
		\drawfoex{v2}{v3}{$0$}{$1$}
		\drawfoex{v1}{v2}{$6$}{$6$}
	\end{tikzpicture}
	\caption{Allocation III}
	\label{f4}
\end{subfigure}
\caption{Unique Survivor in Each Outcome}
\label{contrast}
\end{figure}

This example shows three equilibrium outcomes, each of which has respectively country 1, 2 and 3 as \emph{the only survivor}. These three equilibrium outcomes are $[\text{safe}, \text{unsafe}, \text{unsafe}]$, $[\text{unsafe}, \text{safe}, \text{unsafe}]$, and $[\text{unsafe}, \text{unsafe}, \text{safe}]$.  Obviously, other countries than the survivor have exhausted all their power in the antagonism with their (other) adversaries.

Motivated by the above example, Theorem~\ref{prop/survive} provides a sufficient condition for the existence of a Nash equilibrium  in the PAG where only a country can survive in an environment where no country has any friend and, in particular, every country is adversary with one another.  To prove this theorem,  the following concept more extensively discussed in \cite{balancing} will be used:

A Nash equilibrium of the PAG, $U^*\in\scr U^*$, is called a {\em balancing equilibrium} if at the equilibrium, every country has to use all its power only on offending its foes, and every country's offense toward every foe is just equal to the offense received from the foe; consequently, every country's total support just balances out its total threats.  
That is, for all $i\in\mathbf n$, there holds 
$u^{*}_{ij} = u_{ji}^{*}$, $j \in \mathcal{A}_{i}$ and $\sum_{j \in \mathcal{A}_{i}}u^{*}_{ij} = p_{i}$. Therefore, $\sigma_{i}(U^{*})  = \tau_{i}(U^{*}) (\text{or}~ x_{i}(U^{*}) = \text{precarious})$.

The following lemma (the proof of which is in \cite{balancing}) provides a condition for the existence of a ``balancing equilibrium''. 



\begin{lemma}
\label{prop/complete-BE}
In the PAG $\Gamma = \{\mathcal{C},  p, \mathcal{U}, \sigma_{i}, \tau_{i}, \mathcal{X}, \preceq\}$ with at least three countries, and all of the countries are adversaries with each other, \footnote{
Lemma~\ref{prop/complete-BE} also holds in the trivial case when there are only two countries.} 
a balancing equilibrium exists if and only if each country's power is no greater than the total power of its adversaries. In other words, suppose $i \in \mathbf{n}$ and $\mathcal{A}_{i} = \mathbf{n} \setminus \{i\}$.  A balancing equilibrium $U^{*} \in \mathcal{U}^{*}$ at which 
$$\forall i \in \mathbf{n}, u_{ii} = 0;$$ $$\forall i \in \mathbf{n}, j \in \mathcal{A}_{i}, u_{ij} = u_{ji}; ~\text{and}~ \sum_{j \in \mathcal{A}_{i}}u_{ij} = p_{i}$$ exists if and only if $$\forall i \in \mathbf{n}, p_{i} \leq \sum_{j \in \mathcal{A}_{i}}p_{j}.$$

\end{lemma}

\begin{theorem}
\label{prop/survive}
In the PAG $\Gamma = \{\mathcal{C},  p, \mathcal{U}, \sigma_{i}, \tau_{i}, \mathcal{X}, \preceq\}$, if all of the countries are adversaries with each other\footnote{
That is, adversary relations constitute a complete graph.} 
and each country's power is strictly less than the total power of all its adversaries,
there exists a Nash equilibrium at which a country and \emph{only} this country can survive (and only this country is safe.) 
In other words,  if for all $i \in \mathbf{n}$, there holds 
$$\mathcal{A}_{i} = \mathbf{n} \setminus \{i\} \;\;\;\;\; 
\text{and} \;\;\;\;\; p_{i} < \sum_{j \in \mathcal{A}_{i}}p_{j},$$ 
then there exists a Nash equilibrium $U^{*} \in \mathcal{U}^{*}$ at which 
$\sigma_{i_0}(U^{*})  > \tau_{i_0}(U^{*})$ and
$\sigma_{j}(U^{*}) < \tau_{j}(U^{*})$ for all $j \in \mathbf{n} \setminus \{i_0\}$, $i_{0} \in \mathbf{n}$
\end{theorem}

\begin{proof}  Given an arbitrary country $i$ in $\mathbf{n}$, the set of adversarial pairs except for those involving $i$ is denoted as $\mathcal{R}_{\mathcal{A}} \setminus \mathcal{A}_{i}$. Note that $\mathcal{R}_{\mathcal{A}} \setminus \mathcal{A}_{i}$ still make up a complete subgraph of $\mathbb{G}$, $\mathbb{G}' = \{\mathbf{n}\setminus \{i\}, \mathcal{E}'\}$.

\begin{enumerate}
\item If there exists a country $j$ in the subgraph $\mathbb{G}'$, that is, $j \in \mathbf{n}\setminus \{i\}$, such that its power is no smaller than that of all other countries (i.e., its adversaries) combined in the subgraph, $$p_{j} > \sum_{k \in \mathcal{A}_{j}\setminus \{i\}}p_{k}.$$  

In this case, a pure strategy Nash equilibrium $U^{*}$ in which only $i$ survives and is safe, i.e., $\sigma_{i}(U^{*}) > \tau_{i}(U^{*})$, can be constructed. 

The construction proceeds with two steps. First, let country $j$ allocate enough to make all of its adversaries other than $i$ unsafe. Technically speaking, construct an $U' = [u_{jk}]_{(n-1) \times (n-1)}$ where there holds $$\forall k \in \mathcal{A}_{j}\setminus \{i\}, u_{jk} > p_{k}$$ and $$\sum_{h \in \mathcal{A}_{k}\setminus \{i\}}u_{kh} = p_{k}$$

Second, let country $i$ allocate enough to make $j$ unsafe. Technically speaking, construct an $U = [u_{ij}]_{n \times n}$ by expanding $U'$ to incorporate the allocations between $i$ and countries in $\mathbf{n}\setminus \{i\}$. Let $u_{ij} > p_{j} - \sum_{k \in \mathcal{A}_{j}\setminus \{i\}}u_{jk}$. 

This is feasible because, as assumed, $p_{i} < \sum_{j \in \mathcal{A}_{i}}$.

Then $$p_{j} - p_{i} \leq \sum_{k \in \mathcal{A}_{j}\setminus \{i\}}p_{k} \leq \sum_{k \in \mathcal{A}_{j}\setminus \{i\}}u_{jk}$$

Rearranging terms, $$p_{i} \geq p_{j} - \sum_{k \in \mathcal{A}_{j}\setminus \{i\}}u_{jk}.$$

 Then a pure strategy equilibrium has been derived such that $\sigma_{i}(U^{*})  > \tau_{i}(U^{*})$ and
$\sigma_{j}(U^{*}) < \tau_{j}(U^{*})$ for all $j \in \mathbf{n}\setminus \{i\}.$

\item If there does not exist a country in $\mathbb{G}'$ such that its power exceeds all other countries in $\mathbb{G}'$.   By lemma~\ref{prop/complete-BE}, a balancing equilibrium $U'$ exists for the PAG of the $n-1$ countries on $\mathbb{G}'$. 

Let it be $U' = [u_{jk}]_{(n-1) \times (n-1)}$, where by definition $$\forall j \in \mathbf{n}', u_{jj} = 0; \forall j, k \in \mathbf{n}', u_{jk} = u_{jk};$$ $$\sum_{k \in \mathcal{A}_{j}\setminus \{i\}} u_{jk} = p_{j}.$$

In this case, a pure strategy Nash equilibrium in which only $i$ survives can be constructed by expanding $U'$ to incorporate the allocations between $i$ and countries in $\mathbf{n}\setminus \{i\}$. $\forall j \in \mathbf{n}\setminus \{i\}$, let  $u_{ij} = \frac{p_{i}}{n-1}$.  Then a pure strategy equilibrium has been derived such that $\sigma_{i_0}(U^{*}) > \tau_{i_0}(U^{*})$ and
$\sigma_{j}(U^{*}) < \tau_{j}(U^{*})$ for all $j \in \mathbf{n}\setminus \{i\}.$

 \end{enumerate}

\end{proof}

\begin{theorem}
\label{prop/survive2}
In the PAG $\Gamma = \{\mathcal{C},  p, \mathcal{U}, \sigma_{i}, \tau_{i}, \mathcal{X}, \preceq\}$, if no country has any friend and all the adversary relations among the countries constitute a bipartite graph,  the necessary condition for the PAG to have a Nash equilibrium at which a country not only survives but also is safe is that 
the power of any adversary of this country is no greater than the total power of all its own adversaries (including this country itself). In other words, in the bipartite environment graph $\mathbb{G}_{E} = \{\mathcal{V}, \mathcal{E}_{E}\}$ with the partition of the node set $\mathcal{V}$ into $\mathcal{L}$ and $\mathcal{R}$, $\mathcal{V} = \mathcal{L} \cup \mathcal{R}$ and $\mathcal{L} \cap \mathcal{R} = \emptyset$. $\forall \{v_{i},v_{j}\} \in \mathcal{E}_{E}$, either $v_{i} \in \mathcal{L}$ and $v_{j} \in \mathcal{R}$ or $v_{i} \in \mathcal{R}$ and $v_{j} \in \mathcal{L}$. All edges in $\mathcal{E}_{E}$ are colored red because $\forall i \in \mathbf{n}$, $\mathcal{F}_{i} = \emptyset$. For the PAG to have a pure strategy Nash equilibrium $U^{*}$ in which given country $i$, $\sigma_{i}(U^{*})  > \tau_{i}(U^{*})$, there must hold that $\forall j \in \mathcal{A}_{i}$, $$p_{j} \leq \sum_{k \in \mathcal{A}_{j}}p_{k},$$.

\end{theorem}

\begin{proof}
The contrapositive is that given country $i$, there exists an adversary of itself $j$ whose total power exceeds that of all $j$'s adversaries. 

It follows that 
$$\sigma_{j}(U^{*}) \geq p_{j} > \sum_{k \in \mathcal{A}_{j}}p_{k} \geq \tau_{j}(U^{*}).$$
which implies that country $j$ is always safe in any equilibrium ; consequently, country $i$ is always unsafe or precarious in any equilibrium.

Equivalently,  the necessary condition for the PAG to have a pure strategy Nash equilibrium $U^{*}$ such that $\sigma_{i}(U^{*})  > \tau_{i}(U^{*})$,  which is that $\forall j \in \mathcal{A}_{i}$, $$p_{j} \leq \sum_{k \in \mathcal{A}_{j}}p_{k}$$ is thus proven.

\end{proof}

The necessary condition stated in Theorem~\ref{prop/survive2} is insufficient for the proposition to hold (see Example 3), and Theorem~\ref{prop/survive3} states a sufficient condition.

\subsubsection*{Example 3}
The networked environment in which the PAG takes place is characterized by the following parameters:
\begin{enumerate}
\item Set of countries (labels): $\mathbf{n} = \{1,2,3,4\}$
\item Countries' total power: $p = [p_{1},p_{2},p_{3},p_{4}] = [4,5,6,5]$.

\item Countries' relations: $\mathcal{A}_{1} =\mathcal{A}_{2} = \{3, 4\}$, and all the other pairs have no relations. 

\item Countries' preferences: Assume the two axioms. 

\end{enumerate}

This example shows that for country 1 or 2, the total power of its adversaries, $p_{3} + p_{4}$, is smaller than that of their adversaries, $p_{1} + p_{2}$. But neither country 1 nor country 2 will be safe in any equilibrium. 

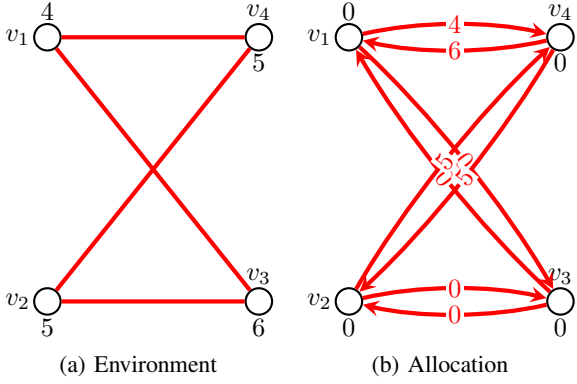
\begin{figure}[htbp]

	\begin{subfigure}[t]{0.45\linewidth}
	\centering
	\begin{tikzpicture}[x=4em,y=-4em]
		\drawnodex{0,0}{v1}{left}{$v_{1}$}{above}{$4$}
		\drawnodex{0,2.5}{v2}{left}{$v_{2}$}{below}{$5$}
		\drawnodex{2,2.5}{v3}{above}{$v_{3}$}{below}{$6$}
		\drawnodex{2,0}{v4}{above}{$v_{4}$}{below}{$5$}
		\drawrededge{v1}{v3}
		\drawrededge{v1}{v4}
		\drawrededge{v2}{v3}
		\drawrededge{v2}{v4}
	\end{tikzpicture}\qquad
	\caption{Environment}
	\label{fail1}
	\end{subfigure}
   \begin{subfigure}[t]{0.45\linewidth}
	\centering
	\begin{tikzpicture}[x=4em,y=-4em]
	           \drawnodex{0,0}{v1}{left}{$v_{1}$}{above}{$0$}
		\drawnodex{0,2.5}{v2}{left}{$v_{2}$}{below}{$0$}
		\drawnodex{2,2.5}{v3}{above}{$v_{3}$}{below}{$0$}
		\drawnodex{2,0}{v4}{above}{$v_{4}$}{below}{$0$}
	
		\drawfoex{v1}{v3}{$0$}{$0$}
		\drawfoex{v1}{v4}{$4$}{$6$}
		\drawfoex{v2}{v3}{$0$}{$0$}
	           \drawfoex{v2}{v4}{$5$}{$5$}
	\end{tikzpicture}
	\caption{Allocation}
		\label{fail2}
	\end{subfigure}
    \caption{Impossibility of Safety for v1 and v2}
    \label{fail}
\end{figure}

\begin{theorem}
\label{prop/survive3}
In the PAG $\Gamma = \{\mathcal{C},  p, \mathcal{U}, \sigma_{i}, \tau_{i}, \mathcal{X}, \preceq\}$, if no country has any friend and all the adversary relations constitute a bipartite graph,  a sufficient condition for the PAG to have a Nash equilibrium in which a country not only survives but also is safe is that 
the power of any adversary of this country is no greater than the total power of its own adversaries (which will include this country), and the total power of these adversaries is no greater than the total power of the adversaries of themselves. 
In other words, in the bipartite environment graph $\mathbb{G}_{E} = \{\mathcal{V}, \mathcal{E}_{E}\}$ with the partition of the node set $\mathcal{V}$ into $\mathcal{L}$ and $\mathcal{R}$, $\mathcal{V} = \mathcal{L} \cup \mathcal{R}$ and $\mathcal{L} \cap \mathcal{R} = \emptyset$. $\forall \{v_{i},v_{j}\} \in \mathcal{E}_{E}$, either $v_{i} \in \mathcal{L}$ and $v_{j} \in \mathcal{R}$ or $v_{i} \in \mathcal{R}$ and $v_{j} \in \mathcal{L}$. All edges in $\mathcal{E}_{E}$ are colored red because $\forall i \in \mathbf{n}$, $\mathcal{F}_{i} = \emptyset$. If for country $i$, there holds that $\forall j \in \mathcal{A}_{i}$, $$p_{j} \leq \sum_{k \in \mathcal{A}_{j}}p_{k}$$ and 
$$\sum\limits_{i \in \mathcal{A}_{i}}p_{j} < \sum\limits_{k \in\mathcal{A}_{\mathcal{A}_{i}}}p_{k}, \;\;\; \text{where} \;\;\; \mathcal{A}_{\mathcal{A}_{i}} = \bigcup\limits_{j \in \mathcal{A}_{i}} \mathcal{A}_{j},$$ (simply, the adversaries of $i$'s adversaries),
then the PAG has a pure strategy Nash equilibrium $U^{*}$ in which $\sigma_{i}(U^{*}) > \tau_{i}(U^{*})$.

\end{theorem}

 A variation of the algorithm used for constructing a pure strategy Nash equilibrium for the PAG in \cite{allocation} will be used for constructing an equilibrium with the above particular prediction. Specifically, the algorithm is first applied only on the adversarial pairs without those involving country $i$, which proceeds as below:
 
\noindent {\it The Algorithm.}  Let $q$ be the number of pairs in $\mathcal{R}_{a}\setminus \{\{i,j\}|j \in \mathcal{A}_{i}\}$, and $\mathbf{q} = \{1,2, ..., q\}$ be the set of distinct labels for elements in this set. A bijection $\gamma: \mathcal{R}_{a} \to \mathbf{q}$ determines an ordering of the set of the adversarial pairs except those involving $i$, $\mathcal{R}_{a}\setminus \{\{i,j\}|j \in \mathcal{A}_{i}\}$, with $\{j,h\}$ being the $\gamma(\{j,h\})$-th term in the ordering. Let $z(k)$ be the vector of countries' remaining power after the $k$-th recursion, where $z_{i}(k)$ is the $i$-th entry in $z(k)$ denoting country $i$'s remaining power. 

Consider the recursion, $$
z(k) =  z(k-1) -\text{min}\{z_{j}(k-1), z_{h}(k-1)\}(e_{j} + e_{h})$$
$$U(k) = \text{diag}\{z(k-1) - \text{min}\{z_{j}(k-1), z_{h}(k-1)\}(e_{j} + e_{h})\} + $$  $$\text{min}\{z_{j}(k-1), z_{h}(k-1)\}(e_{j}e^{T}_{h} + e_{h}e^{T}_{j})$$where  $k \in \mathbf{q}$, $U(k) \in \mathbb{R}^{n \times n}$, $U(0) = \text{diag}\{z(0)\} = \text{diag}\{z_{1}(0), z_{2}(0), ..., z_{n}(0)\}$, and $\{j,h\} = \gamma^{-1}(k)$.


Update $u_{ij}(q) = z_{j}(q) + \epsilon$, $j \in \mathcal{A}_{i}$, subject to $i$'s total power constraint, $\sum_{j \in \mathcal{A}_{i}}z_{j}(q) + \epsilon = p_{i}$.

\begin{proof}  At the end of the algorithm, $U(q)$ as returned is a Nash equilibrium with none having any incentives to deviate: 

\begin{enumerate}
\item For $i$, $\sigma_{i}(U(q)) > \tau_{i}(U(q))$. For any $j \in \mathcal{A}_{i}$, $\sigma_{i}(U(q)) < \tau_{i}(U(q))$. By the axioms, it achieves the best possible power allocation outcome induced by $U(q)$ and therefore has no incentives to deviate. 

\item For any adversary of $i$, $j \in \mathcal{A}_{i}$, $\sigma_{j}(U(q)) < \tau_{j}(U(q))$. They cannot deviate in any way to strictly improve the power allocation outcome. 
\item For any other country, $j \in \mathbf{n} \setminus \{i\}\cup \mathcal{A}_{i}$, $\sigma_{j}(U(q)) \geq \tau_{j}(U(q))$ and $\forall k \in \mathcal{A}_{j}$, $\sigma_{j}(U(q)) \leq \tau_{j}(U(q))$. They also achieve the best possible power allocation outcome, and thus do not have incentives to deviate. 

\end{enumerate}

\end{proof}

\section{Unique Predictions for Survival}

\subsection{Domination and Protectorate}

In certain environments, the PAG can have unique predictions for countries' survival. As part of formalizing those conditions for these environments,  the notions of ``domination'' and ``protectorate'' will first be introduced. A concept of ``domination-protectorate cover''  will then be discussed and used for establishing the conditions for the PAG to have these unique predictions for countries' survival. 

\subsubsection*{Domination} In an environment graph $\mathbb{G}_{E} = (\mathcal{V}, \mathcal{E}_{E})$ that represent the set of countries and their relations in an environment, if for country $i \in \mathbf{n}$, there holds $$p_{i} \geq \sum_{j \in \mathcal{A}_{i}} p_{j} + \sum_{k \in \bigcup\limits_{j \in \mathcal{A}_{i}}\mathcal{F}_{j}}p_{k},$$ 
we call the set $$\mathcal{D}_{i} = \{i\} \cup \mathcal{A}_{i} \cup \bigcup\limits_{j \in \mathcal{A}_{i}}\mathcal{F}_{j}$$
country $i$'s \emph{domination}, which includes itself, its adversaries, and the friends of its adversaries.

\subsubsection*{Protectorate} In an environment graph $\mathbb{G}_{E} = (\mathcal{V}, \mathcal{E}_{E})$ that represent the set of countries and their relations in an environment, let $$\Xi_{i} = \left\{j \in \mathcal{F}_{i}: p_{j} < \sum\limits_{k \in \mathcal{A}_{j}}p_{k}\right\}$$ be the friends of country $i \in \mathbf{n}$ whose total power is smaller than that of its adversaries and $$\Theta_{i} = \bigcup\limits_{j \in \Xi_{i}}\mathcal{A}_{j}$$ be the set of adversaries of this particular set of friends. If $$p_{i} + \sum\limits_{j \in \Xi_{i}}p_{i} \geq \sum\limits_{j \in \mathcal{A}_{i} \cup \Theta_{i}}p_{j},$$ 
we call the set 
$$\mathcal{P}_{i} =  \mathcal{F}_{i} \cup \{i\}$$ 
country $i$'s \emph{protectorate}, which includes itself and all its friends who can defend themselves either with their own power or with that of $i$.

Since environments without any friend relations are degenerate versions of those with both adversary and friend relations, the definitions of \emph{domination} would be essentially the same in both cases, where the only difference is that a country's \emph{domination} in the latter case would cover the friends of its foes. The definition of \emph{protectorate} in environments with both adversary and friend relations would actually be the same with that of \emph{domination} in environments with only adversary relations.

\subsubsection*{Domination-Protectorate Cover}  In an environment graph $\mathbb{G}_{E} = (\mathcal{V}, \mathcal{E}_{E})$, the collection of dominations and protectorates 
$$\mathcal{Q} = \bigcup\limits_{i \in \mathbf{n}}(\mathcal{D}_{i} \cup \mathcal{P}_{i})$$ 
is a domination-protectorate cover of graph $\mathbb{G}_{E}$.

If the domination-protectorate cover spans the whole graph, which means 
$$\mathcal{Q} = \bigcup\limits_{i \in \mathbf{n}}(\mathcal{D}_{i} \cup \mathcal{P}_{i}) = \mathbf{n},$$ 
the PAG on this graph will have a unique prediction for a country's survival, with the reasoning being that the prediction for any country's survival in the game can be \emph{locally and uniquely} determined within each possible domination or protectorate. This reasoning is now illustrated with the below Example 4 and Theorem~\ref{prop/singlenonhobb}.

\subsubsection*{Example 4} The networked environment in which the PAG takes place is characterized by the following parameters:
\begin{enumerate}
\item Set of countries (labels): $\mathbf{n} = \{1,2,3,4\}$
\item Countries' total power: $p = [p_{1},p_{2},p_{3},p_{4}] = [1,2,1,20]$.

\item Countries' relations: $\mathcal{A}_{1} = \{2\}$, $\mathcal{A}_{3} = \{4\}$, $\mathcal{F}_{2} = \{3\}$, and all other possible pairs of countries have no relations.
\item Countries' preferences: Assume the two axioms.

\end{enumerate}

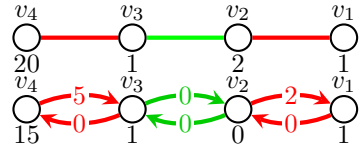
\begin{figure}[ht]
\centering
\begin{tikzpicture}[x=4em,y=-4em][scale=0.25]
	\drawnodex{4,0}{v1}{above}{$v_{1}$}{below}{$1$}
	\drawnodex{3,0}{v2}{above}{$v_{2}$}{below}{$2$}
	\drawnodex{2,0}{v3}{above}{$v_{3}$}{below}{$1$}
	\drawnodex{1,0}{v4}{above}{$v_{4}$}{below}{$20$}
	
	\drawrededge{v1}{v2}
	\drawgreenedge{v2}{v3}
	\drawrededge{v3}{v4}

\end{tikzpicture}

\begin{tikzpicture}[x=4em,y=-4em][scale=0.13]
	\drawnodex{4,0}{v1}{above}{$v_{1}$}{below}{$1$}
	\drawnodex{3,0}{v2}{above}{$v_{2}$}{below}{$0$}
	\drawnodex{2,0}{v3}{above}{$v_{3}$}{below}{$1$}
	\drawnodex{1,0}{v4}{above}{$v_{4}$}{below}{$15$}
	
	\drawfoex{v1}{v2}{$0$}{$2$}
	\drawallyx{v2}{v3}{$0$}{$0$}
	\drawfoex{v3}{v4}{$0$}{$5$} 	\end{tikzpicture}

\caption{D-P Cover Spanning the Graph}
\label{fig/covernhobb}
\end{figure}


Figure~\ref{fig/covernhobb} shows a domination-protectorate cover that spans the whole graph $\mathbb{G}_{E}$. $\mathcal{D}_{2} + \mathcal{D}_{4} = \mathbf{n}$. The respective PAG only has a unique equilibrium class, where countries $2$ and $4$ survive (with their states being $\text{safe}$ here), and countries $1$ and $3$ do not (with their states being $\text{unsafe}$ here).

\begin{theorem}[Unique Prediction for Survival]\label{prop/singlenonhobb} The PAG $\Gamma = \{\mathcal{C},  p, \mathcal{U}, \sigma_{i}, \tau_{i}, \mathcal{X}, \preceq\}$ can \emph{uniquely} predict a country's survival if its domination-protectorate cover spans the whole graph $\mathbb{G}$. In other words, if $\mathcal{Q} = \mathbf{n}$ and there exists a $U^{*} \in \mathcal{U}^{*}$ in which for a country $i \in \mathbf{n}$, $\sigma_{i}(U^{*}) > \tau_{i}(U^{*})$
(or $\sigma_{i}(U^{*}) < \tau_{i}(U^{*})$),  there does not exist $V^{*} \in \mathcal{U}^{*}$ in which $\sigma_{i}(U^{*}) < \tau_{i}(U^{*})$
(or $\sigma_{i}(U^{*}) > \tau_{i}(U^{*})$). 

\end{theorem}

{\em Proof:}
If a domination-protectorate cover spans the whole graph $\mathbb{G}_{E}$, country $i$ is either in a domination, i.e., dominating others or being dominated by others, or is in a protectorate. 
 If country $i$'s domination $\mathcal{D}_{i} \neq \emptyset$ and 
If $i$ is in a protectorate $\mathcal{P}_{i}$, then for all $j \in \mathcal{P}_{i}$, $\sigma_{j}(U^{*}) \geq \tau_{j}(U^{*})$ for any $U^{*} \in \mathcal{U}^{*}$.
Therefore, if there exists a $U^{*} \in \mathcal{U}^{*}$ in which for a country $i \in \mathbf{n}$, $\sigma_{i}(U^{*}) > \tau_{i}(U^{*})$
($\sigma_{i}(U^{*})  < \tau_{i}(U^{*})$),  there does not exist $U^{*} \in \mathcal{U}^{*}$ in which $\sigma_{i}(U^{*}) < \tau_{i}(U^{*})$
($\sigma_{i}(U^{*}) > \tau_{i}(U^{*})$).
Thus, any country's survival can be uniquely determined.
\hfill$\qed$

\section{Conclusion}

This is apparently the first paper to study the countries' survival problem in a networked international environment. One direction of future work is to explore the environments in which a country can not just survive itself but also ``succeed'' in the PAG. This would go beyond the two axioms to a total order of the state space (and a corresponding utility function representation of countries' preferences) in order to rigorously define a country's ``success''.  Another direction is to apply the theory of equilibrium selection to the PAG in our countries' survival problem. For instance, in a PAG with multiple Nash Equilibria, 
can countries always manage to ``select'' those equilibria in which they survive?



\bibliographystyle{unsrt}
\bibliography{alliance,alliance2,alliance3,coalition,nato,colonel}

\end{document}

%% file: steve.tex
\def\send#1#2{\stackrel{#1}{\hbox to #2{\rightarrowfill}}}
\def\-{\!\!\!\!\!-}

 \def\qed{ \rule{.1in}{.1in}}

\def\scr#1{{\cal #1}}

\def\qed{ \rule{.1in}{.1in}}

\newcounter{seqn}[equation]
\def\theseqn{\arabic{equation}\alph{seqn}}

\def\endseqn{\eqno \@seqnnum
$$\ignorespaces}
\def\@seqnnum{(\theseqn)}

\newskip\mcentering \mcentering=0pt plus 1000pt minus 1000pt

\def\meqalignno#1{
\halign to\displaywidth{
    \hbox to 0pt{\kern\displaywidth\llap{$##$}\hss}\tabskip=\mcentering
    &\hfil$\displaystyle{##}$\tabskip=\mcentering
   &&$\displaystyle{{}##}$\hfil\tabskip=\mcentering
    \crcr
    #1\crcr}}

\def\dspace{\multiply\normalbaselineskip 150
		  \divide\normalbaselineskip 100 \normalbaselines
		  \csname @@normalbaselineskip\endcsname\normalbaselineskip}
\def\sspace{\multiply\normalbaselineskip 200
		 \divide\normalbaselineskip 300 \normalbaselines
		 \csname @@normalbaselineskip\endcsname\normalbaselineskip}
\def\sdspace{\multiply\normalbaselineskip 160
		 \divide\normalbaselineskip 150 \normalbaselines
		 \csname @@normalbaselineskip\endcsname\normalbaselineskip}

\def\@{\tilde}

\def\3dot#1{\buildrel\textstyle...\over#1}

